\documentclass[12pt,a4paper]{article}%
\usepackage{authblk}
\usepackage[all,cmtip]{xy}
\usepackage{amssymb}
\usepackage{amsmath}
\usepackage{amsthm}
\usepackage{framed}
\usepackage{comment}
\usepackage{color}
\usepackage{xcolor}
\usepackage{hyperref}
\usepackage[sc]{mathpazo}
\usepackage[T1]{fontenc}
\usepackage{needspace}
\usepackage{tabls}
\usepackage{stmaryrd}
\usepackage{clrscode3e}
\usepackage{clrscode3e}
\usepackage{lipsum}
\usepackage{pgfplots}
\usepackage{float}
\usepackage[latin9]{inputenc}
\usepackage{tikz-cd} 

\newcommand{\N}{\mathbb{N}}
\newcommand{\Z}{\mathbb{Z}}
\newcommand{\Q}{\mathbb{Q}}

\newcommand{\yyy}{\textsc{yes }}
\newcommand{\nnn}{\textsc{no }}
\newcommand{\kfact}{\textsc{k-factor }}
\newcommand{\kfactQ}{\textsc{k-factor}}
\newcommand{\kbal}{\textsc{k-equal-factor }}

\newcommand{\constantterm}{\textsc{constant-term }}

\newcommand{\sumcoeff}{\textsc{sum-of-coefficients }}

\newcommand{\nfdp}{\textsc{natural-reducibility }}

\newcommand{\subsum}{\textsc{subset-sum }}
\newcommand{\subsumQ}{\textsc{subset-sum}}
\newcommand{\subprodQ}{\textsc{subset-product}}
\newcommand{\subprod}{\textsc{subset-product }}
\newcommand{\prodpart}{\textsc{product-partition }}
\newcommand{\prodpartQ}{\textsc{product-partition}}
\newcommand{\subsumprod}{\textsc{subset-sum-of-products }}
\newcommand{\subsumprodQ}{\textsc{subset-sum-of-products}}

\newcommand{\eqct}{\textsc{equal-constant-term }}

\newcommand{\eqsc}{\textsc{equal-sum-of-coefficients }}

\newcommand{\fwsc}{\textsc{factor-with-specific-coefficients }}
\newcommand{\fwscQ}{\textsc{factor-with-specific-coefficients}}

\newcommand{\inputt}{\textbf{Input}}
\newcommand{\outputt}{\textbf{Output}}

\providecommand{\U}[1]{\protect\rule{.1in}{.1in}}
\newcommand{\problem}[1]{\emph{#1}} 

\newcounter{exer}
\newcounter{exera}
\theoremstyle{definition}

\newtheorem{theo}{Theorem}[section]

\newtheorem{lem}[theo]{Lemma}

\newtheorem{oss}[theo]{Observation}

\newtheorem{prop}[theo]{Proposition}

\newtheorem{defi}[theo]{Definition}

\newtheorem{remk}[theo]{Remark}

\newtheorem{coro}[theo]{Corollary}

\newtheorem{conv}[theo]{Convention}

\newtheorem{quest}[theo]{Question}

\newtheorem{warn}[theo]{Warning}

\newtheorem{conj}[theo]{Conjecture}

\newtheorem{exam}[theo]{Example}
\newenvironment{example}[1][]
{\begin{exam}[#1]\begin{leftbar}}
{\end{leftbar}\end{exam}}
\newtheorem{exmp}[exer]{Exercise}

\newtheorem{exetwo}[exera]{Additional exercise}

\let\sumnonlimits\sum
\let\prodnonlimits\prod
\let\cupnonlimits\bigcup
\let\capnonlimits\bigcap
\renewcommand{\sum}{\sumnonlimits\limits}
\renewcommand{\prod}{\prodnonlimits\limits}
\renewcommand{\bigcup}{\cupnonlimits\limits}
\renewcommand{\bigcap}{\capnonlimits\limits}
\excludecomment{verlong}
\includecomment{vershort}
\excludecomment{noncompile}

\begin{document}

\title{Hard to Detect  Factors of Univariate Integer Polynomials}
\author{Alberto Dennunzio}
\affil{\small Dipartimento di Informatica, Sistemistica e Comunicazione,
  Università  degli Studi di Milano-Bicocca,
  alberto.dennunzio@unimib.it
}

\author{Enrico Formenti}
\affil{Universit\'e C\^ote d'Azur, 
enrico.formenti@unice.fr
}

\author{Luciano Margara\footnote{Corresponding author. Address:  Department of Computer Science and Engineering, Via dell'Università 50, 47521 Cesena, Italy}}
\affil{Department of Computer Science and Engineering, University of Bologna, 
luciano.margara@unibo.it
}

\date{\today}

\maketitle

\begin{abstract}
We investigate the computational complexity of deciding whether  a given 
univariate integer polynomial $p(x)$ has a factor $q(x)$ satisfying specific
additional constraints. When the only constraint imposed on $q(x)$ is to have a 
degree smaller than the degree of $p(x)$ and greater than zero,
the problem is equivalent to testing the irreducibility of $p(x)$
and then it is solvable in polynomial time. We prove that deciding whether a given 
monic univariate integer polynomial  has factors   
satisfying additional properties 
may lead to NP-complete problems in the strong sense. 

In particular, given any constant value $k\in \Z$,  we prove that it is NP-complete in the strong sense to detect the existence of a factor that returns a 
prescribed value when evaluated at  $x=k$ (Theorem \ref{teo1}) or to detect the existence of a pair of factors 
- whose product is equal to the original polynomial - that return the same value when evaluated at  $x=k$ (Theorem \ref{teo2}).
The list of all the properties we have investigated in this paper  is reported at the end of Section \ref{sec1}.
\end{abstract}


\section{Introduction}\label{sec1}
The idea of decomposing a polynomial into the product of  smaller ones
is definitely not new. 
A huge literature has been devoted to the factorization of polynomials (without claim of exhaustiveness see    \cite{VANHOEIJ2002167,Lenstra82factoringpolynomials,DBLP:conf/latin/Kaltofen92}) 
as well as to the decomposition
of other mathematical objects, e.g.
numbers, matrices, graphs and so on.
The basic idea behind factorization is decomposing a complex object into smaller and easier to analyze pieces.
Properties satisfied by each piece might shed some light on the properties satisfied by the entire object.
As an example, from irreducible factors of a polynomial we can recover valuable information about its roots.
%
%

As far as polynomials are concerned, much attention has been dedicated to the problem of factoring them into irreducible polynomials, i.e into no furtherly factorable elements.
The first polynomial factorization algorithm was published by Theodor Von Schubert in 1793 \cite{Schubert}.
Since then, dozens of papers on the computational complexity of polynomial factorization have been published.
In 1982, Arjen K. Lenstra, Hendric W. Lenstra, and László Lovász \cite{Lenstra82factoringpolynomials} published the first polynomial time algorithm for factoring 
polynomials over $\Q$ and then over $\Z$.

When dealing with the computational complexity of problems whose input is a polynomial, it is crucial to specify the way we represent it. The standard way of representing a polynomial 
$p(x)=c_0+c_1 x+\cdots + c_nx^n$  is by giving the list 
$\langle c_0,c_1,\dots ,c_n \rangle$ of its coefficients. In this case the size of the polynomial is proportional to $n$ and does not depend
on the number of zero coefficients in $p(x)$. The other way to represent a polynomial (called lacunary or sparse representation) consists of the list of nonzero monomials. Lacunary representation may lead to an exponentially
shorter representation of the same polynomial with respect to the standard notation. The computational cost of an algorithm can be polynomially bounded in the standard input  size and, at the same time, exponentially large in the lacunary   input size.
 Testing the irreducibility of lacunary polynomials or computing the greatest common divisor of two lacunary polynomials are NP-hard problems 
\cite{Plaisted77,KarpinskiS99,KaltofenK05}. Computing the irreducible factors of bounded degree of lacunary polynomials can be done in polynomial time  \cite{GRENET2016171} as well as computing the integer roots of lacunary  integer polynomials \cite{CUCKER199921}.

In this paper we will use standard notation for polynomials. This makes our NP-completeness results even stronger.

 Irreducible factorization of polynomials, under some conditions, is unique and any other factorization into not necessarily irreducible elements can be obtained by properly grouping suitable irreducible factors.
In this paper we focus our attention on some particular type of factors, not necessarily irreducible, and on the computational complexity of detecting their existence.
We will prove that some type of factors are hard to detect while some others are not.
In other words, we  show that computing irreducible factors of a polynomial can be much easier than computing other type of factors. It turns out, as expected,  that the boundary between polynomially computable factors and 
"hard to compute" ones is far from being  completely understood.

We wish to emphasize that the main aim of this paper is not to provide technically difficult proofs of long standing open problems but rather to
show a different perspective in dealing with polynomial decomposition problems. 

We face the following general problem.
Given an integer polynomial $p(x)\in \Z[x]$ and some specific property $P$,
decide whether $p(x)$ admits one or more factors  that satisfy $P$.

Here is a list of problems we have analyzed in this paper.

\begin{itemize}
\item[Q1.]   Let $k\in \Z$ be any fixed integer. \\
Given  $p(x)\in \Z[x]$ and $h\in\Z$, decide whether there exists
 a factor $q(x)\in\Z[x]$ of $p(x)$ such that $q(k)=h$. (Theorem \ref{teo1})
\item[Q2.]  
Given  $p(x)\in \Z[x]$ and $h\in\Z$,  decide whether there exists
 a factor $q(x)\in\Z[x]$ of $p(x)$ such that the sum of all the coefficients of $q(x)$ is equal to $h$.
(Corollary \ref{sumcoeffcoro})
\item[Q3.]  
Given  $p(x)\in \Z[x]$ and $h\in\Z$,  decide whether there exists
 a factor $q(x)\in\Z[x]$ of $p(x)$ such that the constant term of $q(x)$ is equal to $h$. (Corollary \ref{teoct})
 \item[Q4.] Let $k\in \Z$ be any fixed integer. \\
Given  $p(x)\in \Z[x]$,  decide whether there exists two factors $q(x),r(x) \in\Z[x]$ of $p(x)$ such that
$p(x)=q(x)\cdot r(x)$ and $q(k)=r(k)$. (Theorem \ref{teo2}) 
\item[Q5.] 
Given  $p(x)\in \Z[x]$,  decide whether there exists two factors $q(x),r(x) \in\Z[x]$ of $p(x)$ such that
$p(x)=q(x)\cdot r(x)$ and the sum of all the coefficients of $q(x)$ is equal to 
the sum of all the coefficients of $r(x)$. (Corollary \ref{balancedconsttermcoro})
\item[Q6.] 
Given  $p(x)\in \Z[x]$,  decide whether there exists two factors $q(x),r(x) \in\Z[x]$ of $p(x)$ such that
$p(x)=q(x)\cdot r(x)$ and constant term of $q(x)$ is equal to 
the constant term of $r(x)$. (Corollary \ref{balancedsumcoeffcoro})
\item[Q7.] 
Given  $p(x)\in \Z[x]$,  decide whether there exists two factors $q(x),r(x) \in\N[x]$ of $p(x)$ such that
$p(x)=q(x)\cdot r(x)$. (Question \ref{op1})
\item[Q8.] Given  $p(x)\in \Z[x]$ and $h,m\in\Z$,  decide whether there exists
 a factor $q(x)\in\Z[x]$ of $p(x)$ such that the coefficient of the monomial with degree $m$ in $q(x)$ is equal $h$.
 (Question \ref{op2})
\end{itemize}

The rest of the paper is organized as follows. 
In Section \ref{sec2} we give some basic definitions and known results.
In Section \ref{sec3} we prove that  problems from $Q1$ to $Q6$ are NP-complete in the strong sense. 
In Section \ref{sec4} we introduce and discuss open questions $Q7$ and $Q8$. Section \ref{sec5} contains conclusions.

\section{Definitions and Known Results}\label{sec2}

Let $\Z$ denote the set of integer numbers and $\Z[x]$ the set of integer polynomials (polynomials with 
coefficients in $\Z$). Given two integer polynomials $p(x)$ and $q(x)$ we say that 
$q(x)$ divides $p(x)$ (we write $q(x)|p(x)$) if and only if there exists an integer polynomial $r(x)$ such that 
$p(x)=q(x)\cdot r(x)$. 
The degree of a polynomial $p(x)=c_0+c_1x+\cdots + c_n x^n$ with $c_n\neq 0$ 
(denoted by $deg(p(x)))$  is $n$.
Given an integer polynomial $p(x)$ we say that $q(x) \in  \Z[x]$ is a factor of $p(x)$ if and only if
$q(x) | p(x)$.
An integer polynomial with degree $n$ is reducible if and only if it admits a factor $q(x)$ such that 
 $0<deg(q(x))<n$.
It is irreducible otherwise.
In the rest of the paper we will only consider  monic integer polynomials, i.e., integer polynomials whose leading coefficient (coefficient of the highest degree monomial) is equal to $1$. 

We now introduce some well known NP-complete computational problems that we will use for our reductions.

\begin{defi}[\subsumQ]\label{subsum}
Given $n+1$ positive integers $\langle a_1,\dots,a_n,t\rangle $ 
decide whether there exist $I \subseteq \{1,\dots,n \}$ such that
$$\sum_{i\in I} a_i   =t$$
\end{defi}

\begin{defi}[\subprodQ]\label{sp}
Given $n+1$ positive integers $\langle a_1,\dots,a_n,t\rangle $  
decide whether there exist $I \subseteq \{1,\dots,n \}$ such that
$$\prod_{i\in I} a_i   =t$$
\end{defi}

\begin{defi}[\prodpartQ]\label{pp}
Given $n$ positive integers $\langle a_1,\dots,a_n\rangle $  
decide whether there exists a partition of the set $ \{1,\dots,n \}$
into two nonempty subsets $I$ and $J$ such that 
$$\prod_{i\in I} a_i   =\prod_{j\in J} a_j$$
\end{defi}

The \subsum problem  (problem [SP13], page 224 in \cite{garey1979computers}) has been proved to be NP-complete in \cite{Kar72}. It is solvable in pseudo-polynomial time.
The \subprod problem  (problem [SP14], page 225 in \cite{garey1979computers}) has been proved to be NP-complete in the strong sense in \cite{garey1979computers,yao78}. 
The \prodpart problem has been proved to be NP-complete in the strong sense in \cite{NG2010601}. 
A problem is said to be NP-complete in the strong sense, if it remains NP-complete even when all of its numerical parameters are bounded by a polynomial in the length of the input (see \cite{GJ78} for details).


\section{Hard to detect factors}\label{sec3}
In this section we prove that problems from $Q1$ to $Q6$ are NP-complete in the strong sense.

The following observation completely characterizes any factor of a 
monic univariate integer polynomial with integer roots.

\begin{oss}\label{factorform}
Let $n\geq 1$ and $a_1,\dots, a_n$ be $n$ integers.
Let $p(x)\in \Z[x] $ be the following monic univariate integer polynomial
$$p(x) = \prod_{i=1}^n (x-a_i)$$
An integer polynomial $q(x)\in \Z[x]$, $q(x) \neq 1$,  is a factor of $p(x)$ if and only if 
$$\exists I \subseteq \{1,\dots,n \}: \   q(x) = \prod_{i\in I} (x-a_i)$$
\end{oss}


\begin{oss}\label{compcoeff}
Computing all the coefficients of an integer polynomial 
$$p(x)=\prod_{i=1}^n (x-a_i)$$ takes $O(n^2)$ operations.
\end{oss}
\begin{proof}
Let $k\in \{1,\dots,n-1\}$.
 Let $p_k(x)=\prod_{i=1}^k (x-a_i)$. It is easy to verify that the degree of $p_k(x)$ is equal to $k$ and then the number of coefficients of  $p_k(x)$ is at most $k+1$.
Computing $p_{k+1}(x)$ from $p_k(x)$ takes $O(k)$ operations. Then computing $p_{n}(x)=p(x)$  takes 
$O(n^2)$ operations. 
\end{proof}

\begin{defi}[\kfact problem]\label{kfactproblem}
Let $k\in \Z$ be any fixed integer. The  \kfact  problem is defined as follows.
\begin{tabbing}
\=\inputt: $p(x)\in\Z[x]$ and  $h\in \Z$ \\
\> \outputt:  \= - \yyy \= if $p(x)$ has a factor $q(x)\in\Z[x]$ such that $q(k)=h$, \\
\> \> - \nnn otherwise.
\end{tabbing}
\end{defi}

\begin{theo}\label{teo1}
For any fixed $k\in\Z$, \kfact problem is NP-complete in the strong sense.
\end{theo}
\begin{proof}
Let $k\in \Z$ be any fixed integer.
We reduce the  \subprod problem (Definition \ref{sp})  to the \kfact problem.\\
Let  $\langle a_1,\dots,a_n,t \rangle$ be any instance of \subprodQ. Let 
$b_i=a_i-k$ for $i=1,\dots,n$ and 
$$p(x)= \prod_{i=1}^n (x+b_i) $$
Let $\langle p(x),t\rangle$  be the corresponding instance of \kfactQ. 

$\langle p(x),t\rangle$ is a \yyy instance of \kfact if and only if there exists a factor
 $q(x)\in \Z[x]$ of $p(x)$ such that $q(k)=t$. Or equivalently, by Observation \ref{factorform}, if and only if
\begin{equation}\label{eq:1000}
\exists I \subseteq \{1,\dots,n \}: \quad q(x) = \prod_{i\in I} (x+b_i)\text{ and } q(k)=t
\end{equation}
Equation (\ref{eq:1000}) is true if and only if 
\begin{equation}\label{eq:200}
\exists I \subseteq \{1,\dots,n \}: \quad \prod_{i\in I} (k+b_i)=t
\end{equation}
Since
\begin{eqnarray*}
\prod_{i\in I} (k+b_i) 
&=& \prod_{i\in I} (k+a_i-k)\\
&=& \prod_{i\in I} a_i
\end{eqnarray*}
we conclude that $q(k)=t$ if and only if $\prod_{i\in I} a_i=t$. This is true
if and only if $\langle a_1,\dots,a_n,t \rangle$ is a \yyy instance of \subprodQ.

\end{proof}

\begin{oss}\label{iroots}
It is easy to check (directly from the proof of Theorem \ref{teo1})  that the \kfact problem remains NP-complete in the strong sense even if 
we restrict the set of input polynomials to 
 monic polynomials with all integer roots.
\end{oss}

\begin{defi}[\sumcoeff problem]\label{sumc}
 The  \sumcoeff  problem is defined as follows.

\begin{tabbing}
\=\inputt: $p(x)\in\Z[x]$ and  $s\in \Z$ \\
\> \outputt:  \= - \yyy \= if $p(x)$ has a factor $q(x)\in\Z[x]$ such that \\
\> \> \>the sum 
of all the coefficients of $q(x)$ is equal to $s$, \\
\> \> - \nnn otherwise.
\end{tabbing}
\end{defi}

\begin{coro}\label{sumcoeffcoro}
The  \sumcoeff problem is  NP-complete in the strong sense. 
\end{coro}
\begin{proof}
We prove this result as a Corollary of Theorem \ref{teo1}. 
We reduce the \kfact problem with $k=1$ (NP-complete in the strong sense by Theorem \ref{teo1})
to the  \sumcoeff problem.

Let  $q(x)$ be any factor of $p(x)$.
Since $q(1)$ is equal to the sum of all the coefficients of $q(x)$ we conclude that $q(1)=h$ if and only if the 
sum of all the coefficients of
$q(x)$ is equal to $h$.
\end{proof}

\begin{defi}\label{ctproblem}
 The  \constantterm  problem is defined as follows.

\begin{tabbing}
\=\inputt: $p(x)\in\Z[x]$ and  $t\in \Z$ \\
\> \outputt:  \= - \yyy \= if $p(x)$ has a factor $q(x)\in\Z[x]$ such that \\
\> \> \>the constant term of $q(x)$ is equal to $t$, \\
\> \> - \nnn otherwise.
\end{tabbing}

\end{defi}

\begin{coro}\label{teoct}
The \constantterm problem is  NP-complete in the strong sense. 
\end{coro}
\begin{proof}
We prove this result as a Corollary of Theorem \ref{teo1}. 
We reduce the \kfact problem with $k=0$ (NP-complete in the strong sense by Theorem \ref{teo1})
to the  \constantterm problem.

Let  $q(x)$ be any factor of $p(x)$.
Since $q(0)$ is equal to the constant term of $q(x)$ we conclude that $q(0)=h$ if and only if the 
constant term of $q(x)$ is equal to $h$.
\end{proof}

%
%
%

\begin{defi}[\kbal problem]\label{qrth}
Let $k\in \Z$ be  any fixed integer. The problem \kbal  is defined as follows.
\begin{tabbing}
\=\inputt: $p(x)\in\Z[x]$ \\
\> \outputt: \= - \yyy \= if $p(x)$ has two factors $q(x),r(x) \in\Z[x]$ such that \\
\> \> \> $p(x)=q(x)\cdot r(x)$ and $q(k)=r(k)$ \\
\> \> - \nnn otherwise.
\end{tabbing}
\end{defi}

\begin{theo}\label{teo2}
 \kbal is  NP-complete in the strong sense.
\end{theo}
\begin{proof}
Let $k\in \Z$ be any fixed integer.
We reduce the  \prodpart problem (Definition \ref{pp})  to the \kbal problem.\\
Let  $\langle a_1,\dots,a_n \rangle$ be any instance of Product Partition. Let 
$b_i=a_i-k$ for $i=1,\dots,n$ and 
$$p(x)=\prod_{i=1}^n (x+b_i)$$
We now prove that $\langle p(x) \rangle$ is a \yyy instance for \kbal if and only if 
$\langle a_1,\dots,a_n \rangle$ is a \yyy instance for \prodpartQ.

$\langle p(x) \rangle$ is a \yyy instance of \kbal if and only if 
$p(x)$ has two factors $q(x),r(x) \in\Z[x]$ such that 
$p(x)=q(x)\cdot r(x)$ and $q(k)=r(k)$.

Or equivalently, by Observation \ref{factorform}, if and only if
the set $\{1,\dots,n \}$ can be partitioned into two nonempty subsets  $I$ and $J$
such that
\begin{equation}\label{eq:1}
\prod_{i\in I} (k+b_i)=\prod_{j\in J} (k+b_j)
\end{equation}
Since $b_i=a_i-k$ for $i=1,\dots,n$,  Equation \ref{eq:1} can be rewritten as follows
\begin{equation}
\prod_{i\in I} (k+a_i-k)=\prod_{j\in J} (k+a_j-k)
\end{equation}
and then
\begin{equation}\label{eq:2}
\prod_{i\in I} a_i=\prod_{j\in J} a_j
\end{equation}
Equation \ref{eq:2} holds if and only if $\langle a_1,\dots,a_n \rangle$ is a \yyy instance for  \prodpartQ.
\end{proof}

\begin{defi}[\eqct problem]\label{balancedsumcoro}
The  \eqct problem is defined as follows.
\begin{tabbing}
\=\inputt: $p(x)\in\Z[x]$ \\
\> \outputt:  \= - \yyy \= if there exist $q(x),r(x)\in\Z[x]$ such that $q(x)\cdot r(x)=p(x)$ and  \\
\> \> \>the constant term of $q(x)$ is equal to the constant term of $r(x)$, \\
\> \> - \nnn otherwise.
\end{tabbing}

\end{defi}

\begin{coro}\label{balancedconsttermcoro}
The \eqct problem is strongly NP-complete. 
\end{coro}
\begin{proof}
Since the constant term of any polynomial $p(x)$ is equal to $p(0)$, the proof of this Corollary 
follows from Theorem \ref{teo2} setting $k=0$.
\end{proof}

\begin{defi}[\eqsc problem]\label{coro2}
The  \eqsc problem is defined as follows.
\begin{tabbing}
\=\inputt: $p(x)\in\Z[x]$ \\
\> \outputt:  \= - \yyy \= if there exist $q(x),r(x)\in\Z[x]$ such that $q(x)\cdot r(x)=p(x)$ and  \\
\> \> \>the sum of all the coefficients of $q(x)$ is equal to \\
\> \> \>the sum of all the coefficients of $r(x)$, \\
\> \> - \nnn otherwise.
\end{tabbing}
\end{defi}

\begin{coro}\label{balancedsumcoeffcoro}
The \eqsc problem is strongly NP-complete. 
\end{coro}
\begin{proof}
Since sum of the coefficients  of any polynomial $p(x)$ is equal to $p(1)$, the proof of this Corollary 
follows from Theorem \ref{teo2} setting $k=1$.
\end{proof}


\section{Open Questions}\label{sec4}

\subsection{Natural factors detection problem}
Let $\N$ denote the set of natural numbers (positive integer numbers) and $\N[x]$ the set 
of integer polynomials (polynomials with 
coefficients in $\N$). $\N[x]$ with the usual sum and product operations is a commutative semiring.
In fact, this is the free commutative semiring on a single generator $\{x\}$.

\begin{defi}\label{nfdef}
The  \nfdp  problem is defined as follows.
\begin{tabbing}
\=\inputt: $p(x)\in\Z[x]$ \\
\> \outputt: \= - \yyy \= if $p(x)$ has two factors $q(x),r(x) \in\N[x]$ such that \\
\> \> \> $p(x)=q(x)\cdot r(x)$, \\
\> \> - \nnn otherwise.
\end{tabbing}
\end{defi}

\begin{quest}\label{op1}
Is the \nfdp  problem NP-complete ?
\end{quest}

The following example shows a polynomial that is irreducible when considered as an element of $\N[x]$ and reducible when considered as an element of  $\Z[x]$.

\begin{exam}
Let $p(x)=1+x^3$. 
The complete factorization of $p(x)$ in $\Z[x]$ is $p(x)=(1+x)(1-x+x^2)$ while $p(x)$ is irreducible in $\N[x]$.
\end{exam}

In the next example we show that the prime factorization of integer polynomials in $\N[x]$ is not unique.

\begin{exam}
Let $p(x)=1 + x + x^2 + x^3 + x^4 + x^5$. 
The complete factorization of $p(x)$ in $\Z[x]$ is $p(x)=(1 + x) (1 - x + x^2) (1 + x + x^2)$.
Since $(1 + x) (1 - x + x^2)\in \N[x]$ and $(1 - x + x^2) (1 + x + x^2)\in \N[x]$,  then we have two distinct factorizations of $p(x)$ in $\N[x]$.
\begin{eqnarray*}
p(x) &=& (1 + x) (1 + x^2 + x^4)\\
	&=& (1 + x^3) (1 + x + x^2)
\end{eqnarray*}
\end{exam}

Our conjecture is that the \nfdp   problem is NP-complete but we have not been able
to prove it.

\subsection{Factors with specific coefficients detection problem}
 
Let $p(x)=c_0+c_{1}x+\cdots + c_n x^n$ be any integer polynomial. We denote 
by $coef(p(x),m)$, $0\leq m \leq n$, the coefficient $c_{m}$. For values of $m$ outside the interval
$[0,n]$,  $coef(p(x),m)$ is equal to $0$.
According to this definition,
$coef(p(x),n)=c_n$ is the coefficient of the monomial in $p(x)$ with maximum degree (for monic polynomials is always equal to $1$) and
$coef(p(x),0)=c_0$ is the constant term of $p(x)$.

The factor  with specific coefficients detection problem is defined as follows.

 \begin{defi}[\fwscQ] \label{fwscoef}
Let $m\geq 0$ be any fixed integer.
\begin{tabbing}
\=\inputt: $p(x)\in\Z[x]$ and $h\in\Z$\\
\> \outputt: \= - \yyy \= if $p(x)$ has a factor $q(x) \in\Z[x]$ such that \\
\>\>\> $coef(q(x),m)=h$,  \\
\> \> - \nnn otherwise.
\end{tabbing}
\end{defi}
By Corollary \ref{teoct}, we have that for $m=0$, \fwsc problem is NP-complete in the strong sense.
 In fact, when $m=0$ the problem is equivalent to the \constantterm problem.
 
We now define a problem that is a sort of combination of 
\subsum  and \subprod problems.

\begin{defi}[\subsumprodQ]\label{sspp}
Let $k$ be any fixed integer.
Given $n+1$ positive integers $a_1,\dots,a_n,t $ 
decide whether there exist $I \subseteq \{1,\dots,n \}$ such that

\begin{equation}\label{v1}
\sum_{\substack{i_1< i_2<\cdots < i_k  \\ i_j \in I,\  1\leq j \leq k}}^{} 
\left(\prod_{j=1}^{k} a_{i_j} \right)=t
\end{equation}
\end{defi}

Note that for $k=1$ the \subsumprod  problem is nothing but the \subsum problem (Definition \ref{subsum}) and then it is NP-complete. 

In the following theorem we prove that if the \fwsc  detection problem
is not  easier than  the \subsumprod problem.

To this extent we recall the Vieta's formulas (customized for 
monic polynomials over the integers with integer roots) that relate the roots of a polynomial to its coefficients.
\begin{theo}[Vieta's Formulas for monic polynomials over the integers with integer roots]
Let 
\begin{eqnarray*}
p(x) &=& \prod_{i=1}^n (x+a_i)\\
        &=& x^n + c_{n-1} x^{n-1} + \cdots + c_{1} x + c_{0}
\end{eqnarray*}
with $a_i \in \Z$. 
 We have
\begin{eqnarray*}
\forall k\in \{1,\dots,n \}: \quad c_{n-k}&=& \sum_{1\leq i_1<\cdots < i_k \leq n}^{} 
\left(\prod_{j=1}^{k} a_{i_j} \right) 
\end{eqnarray*}

\end{theo}

\begin{exam}
Let 
\begin{eqnarray*}
P_3(x) &=& (x+a_1)(x+a_2)(x+a_3)\\
        &=& x^3 +(a_1+a_2+a_3) x^{2} + (a_1 a_2+a_1 a_3+a_2 a_3) x + a_1 a_2 a_3
\end{eqnarray*}
\begin{eqnarray*}
P_4(x) &=& (x+a_1)(x+a_2)(x+a_3)(x+a_4)\\
        &=& x^4 +(a_1+a_2+a_3+a_4) x^{3} + (a_1 a_2+a_1 a_3+ a_1 a_4+a_2 a_3+a_2 a_4+a_3 a_4) x^2 \\
        &&+ (a_1 a_2 a_3+a_1 a_2 a_4+a_1 a_3 a_4+a_2 a_3 a_4) x + a_1 a_2 a_3 a_4
\end{eqnarray*}

\end{exam}

\begin{theo}\label{reduction}
The \subsumprod problem  is  reducible to 
the \fwsc detection problem.
\end{theo}
\begin{proof}

Let $k\in \Z$ be any fixed integer.
Let  $\langle a_1,\dots,a_n,t \rangle$ be any instance of \subsumprodQ. 
Let $m=k$, $h=t$ and
$$p(x)= \prod_{i=1}^n (x+a_i) $$

Let $\langle p(x),t\rangle$  be the corresponding instance of \fwscQ. 
We now prove that $\langle a_1,\dots,a_n,t \rangle$ is a \yyy instance of \subsumprod
if and only if $\langle p(x),t\rangle$ is a \yyy instance of \fwscQ.

$\langle p(x),t\rangle$ is a \yyy instance of \fwsc if and only if there exists a factor
 $q(x)\in \Z[x]$ of $p(x)$ such that $coef(q(x),m)=h$. 
 By Observation \ref{factorform}, any factor of $p(x)$ has the form
 
 \begin{equation}\label{eq:1100}
q(x) = \prod_{j=1}^{deg(q(x))} (x+b_{i_j})
\end{equation}
where $I= \{ i_1,\dots, i_{deg(q(x))} \}$ is a suitable subset of $\{ 1,\dots,n \}$.

By Vieta's formulas we know that 
$coef(q(x),m)$ can be written as
\begin{equation}\label{vv1}
\sum_{\substack{i_1< i_2<\cdots < i_m  \\ i_j \in I,\  1\leq j \leq m}}^{} 
\left(\prod_{j=1}^{m} b_{i_j} \right)
\end{equation}
This ends the proof
\end{proof}
We end this section with the following open question.
\begin{quest}\label{op2}
For which values of $m$ (other than the case $m=0$) is the \fwsc  problem NP-complete ?
\end{quest}

\section{Conclusion and further works}\label{sec5}
In this paper we have introduced and analyzed the computational complexity of 
the problem of detecting the existence 
of factors of integer polynomials satisfying specific additional constraints.
Even if detecting the existence of  factors of an integer polynomial
can be done in polynomial time,
it turns out that adding  simple constraints on factors leads to hard to solve variants of the problem. 
We prove that problems $Q1$ to $Q6$ listed at the end of Section \ref{sec1} 
are NP-complete in the strong sense.

Section \ref{sec5} provides some ideas for further works. In particular, we were surprised 
that Question \ref{op1} (to our knowledge) had not already been previously addressed  in the literature 
since the problem of factoring an integer polynomial over $\N$ instead of over $\Z$ seems to us a very natural
and interesting question to investigate.

\bibliographystyle{plain}
\bibliography{b}

\end{document}